\documentclass[10pt]{article}
\usepackage{graphicx}
\usepackage{amsmath}
\usepackage{amssymb}
\usepackage{enumerate}

\usepackage{mathtools,enumitem}
\usepackage{mathrsfs} 
\usepackage{overpic,subcaption}
\usepackage{hyperref}
\usepackage{geometry}
\usepackage{xcolor}

\usepackage[linesnumbered,ruled,vlined]{algorithm2e}

\newcommand{\nca}{\texttt{nca}}
\newcommand{\lungh}{w}

\newcommand{\dist}{\mathrm{dist}}

\newtheorem{theorem}{Theorem}
\newtheorem{property}[theorem]{Property}
\newtheorem{lemma}[theorem]{Lemma}
\newtheorem{corollary}[theorem]{Corollary}
\newtheorem{remark}[theorem]{Remark}

\newtheorem{definition}[theorem]{Definition}
\newcommand{\qed}{\hfill \ensuremath{\Box}}

\newenvironment{proof}{\vspace{1ex}\noindent{\bf Proof.}\hspace{0.5em}}
	{\hfill\qed\vspace{2ex}}

\newcommand{\Left} {\text{Left}}
\newcommand{\Right} {\text{Right}}

\begin{document}

\title{A Linear Time Algorithm for Computing Max-Flow Vitality\\ in Undirected Unweighted Planar Graphs}

\author{Giorgio Ausiello\footnote{Dipartimento di Ingegneria Informatica, Automatica e Gestionale, Universit\`a di Roma ``La Sapienza'', via Ariosto 25, 00185 Roma, Italy. Email: \texttt{ausiello@diag.uniroma1.it, ribichini@diag.uniroma1.it}.}
\and 
Lorenzo Balzotti\footnote{Dipartimento di Scienze di Base e Applicate per l’Ingegneria, Universit\`a di Roma ``La Sapienza'',
Via Antonio Scarpa, 16, 00161 Roma, Italy. Email: \texttt{lorenzo.balzotti@uniroma1.it}.}
\and
Paolo G. Franciosa\footnote{Dipartimento di Scienze Statistiche, Universit\`a di Roma ``La Sapienza'',
piazzale Aldo Moro 5, 00185 Roma, Italy. Email: \texttt{paolo.franciosa@uniroma1.it, isabella.lari@uniroma1.it}.} 
\and 
Isabella Lari\footnotemark[3]
\and 
Andrea Ribichini\footnotemark[1]
}

\date{}
\maketitle

\begin{abstract}
The vitality of an edge in a graph with respect to the maximum flow between two fixed vertices $s$ and $t$ is defined as the reduction of the maximum flow value caused by the removal of that edge.

The max-flow vitality problem  has already been efficiently solved for $st$-planar graphs but has remained open for general planar graphs. For the first time our result provides an optimal solution for general planar graphs although restricted to the case of unweighted planar graphs.
\end{abstract}

\textbf{keywords}{: maximum flow, minimum cut, vitality, planar graphs, undirected graphs}

\section{Introduction and related work}\label{se:intro}

Given a graph $G$, where each edge $e$ is allowed to carry a maximum amount $c(e)$ of \emph{flow}, and given two special vertices $s,t \in G$, the maximum flow (max-flow, for short) problem consists in determining a flow assignment to each edge so that the total flow from $s$ to $t$ is maximized. The flow assignment is subject to two constraints: the capacity bound for each edge, given by $c(e)$, and the conservation constraint for each vertex $v$ other than $s$ and $t$, stating that the total flow entering into $v$ must equal the total flow exiting from $v$.

A very wide literature has been produced since the 1950's about algorithms for computing the maximum flow value and/or a max-flow assignment for general graphs and restricted graph classes. We refer the reader to~\cite{ahuja-magnanti} for classical results and to~\cite{king-rao,orlin} for recent efficient max-flow algorithms for general graphs.

In the case of  undirected planar graphs, Reif~\cite{reif} proposed a divide and conquer approach for computing a minimum $st$-cut (hence the maximum flow value by the well known duality theorem between min-cut and max-flow) in $O(n\log^2n)$ worst-case time, where $n$ is the number of vertices in the graph. By plugging in the linear time single-source shortest path (SSSP) tree algorithm for planar graphs by Henzinger \emph{et al.}~\cite{henzinger}, this bound can be improved to $O(n\log n)$. The best currently known approach for computing a minimum $st$-cut and a max-flow assignment in planar undirected graphs is due to Italiano \emph{et al.}~\cite{italiano}, and it achieves $O(n\log \log n)$ time by a two phase approach that exploits the algorithm by Hassin and Johnson~\cite{hassin-johnson}, that proposed an $O(n\log^2 n)$ worst-case time for computing a maximum flow assignment through a planar SSSP tree. Also in this case, the time bound can be improved to  $O(n \log n)$ by applying the algorithm in~\cite{henzinger}.
If $G$ is directed and planar, the divide and conquer approach by Reif for the undirected case cannot be applied: for this case, Borradaile and Klein~\cite{borradaile-klein} presented an $O(n\log n)$ time algorithm, based on a repeated search of left-most circulations. In the case of directed planar unweighted graphs, Eisenstat and Klein presented a linear time algorithm~\cite{eisenstat-klein}.

For directed $st$-planar graphs, i.e., graphs admitting a planar embedding with $s$ and $t$ on a same face, Hassin~\cite{hassin} proved that both the max-flow value and a max-flow assignment can be found by computing an SSSP tree in the dual graph. By applying the algorithm by Henzinger \emph{et al.}~\cite{henzinger} these problems can be solved in linear time.

\medskip
The \emph{vitality} of an edge $e$ with respect to max-flow measures the max-flow decrement observed when edge $e$ is removed from the graph.
A survey on vitality with respect to max-flow problems can be found in~\cite{ausiello-franciosa_1}.
In the same paper, it is proven that:
\begin{itemize}
\item the vitality of all edges in a general undirected graph can be computed by solving $O(n)$ max-flow instances, thus giving an overall $O(n^{2}m)$ algorithm by applying the $O(mn)$ max-flow algorithms described in~\cite{king-rao,orlin};
\item for $st$-planar graphs (both directed or undirected) the vitality of all edges can be found in optimal $O(n)$ worst-case time. The same result holds for determining the vitality of all vertices, i.e., the max-flow reduction under the removal of all the edges incident on a same vertex;
\item the problem of determining the max-flow vitality of an edge is at least as hard as computing the max-flow for the graph, both for general graphs and for the restricted class of $st$-planar graphs.
\end{itemize}
The vitality problem is left open in~\cite{ausiello-franciosa_1} for general planar graphs, both directed and undirected.

The above cited paper by Italiano \emph{et al.}~\cite{italiano} introduces a dynamic structure that allows edge insertion and deletion and answers to $st$ max-flow queries in weighted undirected planar graphs. In particular, with a preprocessing of $O(n\log n)$ worst-case time, they build a data structure that allows to compute the vitality of $h$ edges in $O(hn^{2/3}\log n^{8/3})$ worst-case time.

Balzotti and Franciosa \cite{balzotti-franciosa_3} gave an algorithm able to compute a $\delta$ additive approximation of the vitality of all edges with capacity at most $c$ in $O(\frac{c}{\delta}n+n\log\log n)$ worst-case time in weighted planar graphs. If the graph is unweighted, then, by choosing $c$ and $\delta$ both equal to 1, their algorithm computes the exact vitality of all edges in $O(n\log\log n)$ worst-case time.

\medskip
In this paper, we address the vitality problem in the case of unweighted undirected planar graphs. We propose an algorithm that computes the vitality of all edges in $O(n)$ worst-case time and space.

The paper is organised as follows: in Section~\ref{se:defs} we provide some definitions and preliminary considerations, edges with vitality 1 are characterized in Section~\ref{se:planararcs}, while our algorithm is described in Section~\ref{se:algo}. Section~\ref{se:conclusions} presents some final considerations and open problems.

\section{Definitions and preliminaries}
\label{se:defs}
Given a connected undirected  graph $G=(V,E)$ with $n$ vertices, we denote an edge $e=\{i,j\} \in E$ by the shorthand notation $ij$, and we define $\mbox{dist}(u, v)$ as the distance between vertices $u$ and $v$. We write for short $v\in G$ and $e\in G$ in place of $v\in V$ and $e\in E$, respectively.

We assume that each edge $e=ij \in E$ has an associated positive \emph{capacity} $c(e)$. 
Let $s \in V$ and $t \in V$, $s\neq t$, be two fixed vertices.
A \emph{feasible flow} in $G$ assigns to each edge $e=ij \in E$ two real values $x_{ij} \in [0,c(e)]$ and $x_{ji} \in [0,c(e)]$ such that:
$$\sum_{j:ij \in E} x_{ij} = \sum_{j:ij \in E} x_{ji}, \mbox{ for each } i \in V \setminus \{s,t\}.$$
The \emph{flow from $s$ to $t$} under a feasible flow assignment $x$ is defined as
$$F(x) = \sum_{j:sj \in E} x_{sj} - \sum_{j:sj \in E} x_{js}.$$
The $\emph{maximum flow}$ from $s$ to $t$ is the maximum value of $F(x)$ over all feasible flow assignments $x$ and we denote it by $MF$.

An \emph{$st$-cut} is a partition of $V$ into two subsets $S$ and $T$ such that $s \in S$ and $t \in T$. The \emph{capacity of an $st$-cut} is the sum of the capacities of the edges $ij \in E$ 
such that $|S \cap \{i,j\}|=1$ and $|T \cap \{i,j\}|=1$.
The well known Min-Cut Max-Flow theorem \cite{ford-fulkerson_1} states that the maximum flow from $s$ to $t$ is equal to the capacity of a minimum $st$-cut for any weighted graph $G$.

\begin{definition}
The \emph{vitality} $\mbox{vit}(e)$ of an edge $e$ with respect to the maximum flow from $s$ to $t$, according to the general concept of vitality in \cite{dagstuhl}, is defined as the maximum flow in $G$ minus the maximum flow in $G'=(V,E \setminus \{e\})$.
\end{definition}

The \emph{dual} of a planar embedded undirected graph $G$ is an undirected planar multigraph $G^*$, whose vertices correspond to faces of $G$ and such that for each edge $e$ in $G$ there is an edge $e^{*} = \{f^{*},g^{*}\}$   in $G^{*}$, where $f^{*}$ and $g^{*}$ are the vertices corresponding to the two faces $f$ and $g$ adjacent to $e$ in $G$. The length $\lungh(e^{*})$ of $e^{*}$ equals the capacity of $e$. 

A path from $a$ to $b$ is called an $ab$-path. Given two vertices $x,y$ in a simple path $p$ we denote the subpath of $p$ from $x$ to $y$ by $p(x,y)$ (possibly $x=y$). Given an $ab$-path $p$ and a $bc$-path $q$, we define $p\circ q$ as the (possibly not simple) $ac$-path obtained by the union of $p$ and $q$. 

We are given an $ab$-path $\pi$ in a planar embedded graph and a vertex $v$ of $\pi$ different from $a$ and $b$. Fixing an orientation of $\pi$, for example from $a$ to $b$, we can say that every edge $\{v,w\}$ such that $w$ does not belong to $\pi$ \emph{lies to the left} or \emph{lies to the right} of $\pi$.

Given two paths $\pi_{1}, \pi_{2}$ in a planar embedded graph, a \emph{crossing} between $\pi_{1}$ and $\pi_{2}$ is a subpath of $\pi_{1}$ defined by vertices $v_{1}, v_{2}, \ldots, v_{k}$, with $k \geq 3$, such that the subpath $\pi_1(v_2,v_{k-1})=\pi_2(v_2,v_{k-1})$, and an orientation of $\pi_2$ exists such that edge $v_{1} v_{2}$ lies to the  left of $\pi_{2}$ and edge $v_{k-1} v_{k}$ lies to the right of $\pi_{2}$. Vertices $v_{2}, v_{3}, \ldots, v_{k-1}$ are called \emph{internal vertices} of the crossing. 
Fixing an orientation of $\pi_{1}$, we can define each crossing as \emph{down} (resp., \emph{up}), if the first edge in the crossing lies to the left (resp., right) of $\pi_{2}$, where terms ``up'' and ``down'' are only introduced to distinguish the type of crossing. We say $\pi_{1}$ \emph{crosses $\pi_{2}$ $t$ times} if there are $t$ different crossings between  $\pi_{1}$ and $\pi_{2}$. In a similar way, we can define a crossing between a cycle and a path.

We recall standard union and intersection operators on graphs, and for convenience we define the difference of two graphs as follows.

\begin{definition}
Given two undirected graphs $G=(V(G),E(G))$ and $H=(V(H),E(H))$, we define the following operations and relations:
\begin{itemize}\itemsep0em
\item $G\cup H=(V(G)\cup V(H),E(G)\cup E(H))$,
\item $G\cap H=(V(G)\cap V(H),E(G)\cap E(H))$,
\item $G\setminus H=(V(G),E(G)\setminus E(H))$,
\item $H\subseteq G\Longleftrightarrow V(H)\subseteq V(G)$ $\wedge$ $E(H)\subseteq E(G)$.
\end{itemize}
\end{definition}

\section{Characterizing edges with vitality 1 in the dual graph}
\label{se:planararcs}

We first observe that in graphs with integer edge capacities, the maximum flow value always is an integer number~\cite{ahuja-magnanti}, and edge vitality is defined as the difference between two maximum flows; since the vitality of an edge cannot exceed its capacity, we have:

\begin{remark} In an unweighted graph, both directed and undirected, edges can only have vitality 0 or 1.
\end{remark}

\noindent By the Min-Cut Max-Flow theorem, an edge has vitality 1 if and only if it appears in at least one minimum $st$-cut (i.e., it is essential to achieve maximum flow, for all flow assignments).

We take $G$ to be an undirected planar graph. We assume a planar embedding of the graph is fixed, and in the dual graph $G^*$ defined by this embedding we fix a face $f_{s}^{*}$ adjacent to $s$  and a face $f_{t}^{*}$  adjacent to $t$.
A cycle in $G^{*}$ that separates face $f_{s}^{*}$ from face $f_{t}^{*}$ is called an \emph{$st$-separating cycle}. Note that all the following results also hold in the general case of positive edge capacities.

\begin{property}[see~\cite{itai-shiloach} and Propositions 1 and 2 in~\cite{reif}]\label{pro:cycle}
Let $G$ be an undirected planar graph. An $st$-cut (resp., minimum $st$-cut) in $G$ corresponds to a cycle (resp., shortest cycle) in $G^{*}$ that separates face $f_{s}^{*}$ from face $f_{t}^{*}$.
\end{property}

We propose an $O(n)$ time algorithm for computing the vitality of all edges in an unweighted undirected planar graph. Bearing in mind Property~\ref{pro:cycle}, edges with vitality 1 can be found by determining which edges, in the dual graph, belong to at least one shortest $st$-separating cycle. 

Our algorithm proceeds as follows: in a first phase we compute the dual graph $G^*$ of $G$ according to the fixed embedding, then we ``cut'' $G^*$ by doubling a shortest path $\pi$, joining $f^*_s$ to $f^*_t$, obtaining a doubled dual graph $D$. It is known \cite{itai-shiloach} that shortest $st$-separating cycles in $G^*$ correspond to shortest paths in $D$ joining vertices that were doubled when building $D$, and have distance equal to the max-flow value from $s$ to $t$. We show that it is possible to find a shortest path for each such pair in $O(n)$ overall time, finding a set $U$ of non-crossing shortest paths.

In a second phase we slice $D$ along paths in $U$, and we show that each edge belonging to a shortest $st$-separating cycle can be found by exploring a single slice. This consideration allows us to find all edges belonging to a minimum $st$-separating cycle in $O(n)$ worst-case time. Some technicalities are needed in order to deal with edges contained in slices' borders and to handle some degenerate cases, namely slices containing bridges.

Let $\pi$  be a shortest path in $G^{*}$ from $f_{s}^{*}$ to $f_{t}^{*}$. We call an $st$-separating cycle $C$  \emph{multiple-crossing} or \emph{single-crossing} according to the number of crossings between $C$  and $\pi$.

Let $C$ be a multiple-crossing shortest $st$-separating cycle, and let $L $ be the set of all internal vertices of all crossings between $C$ and $\pi$, thus vertices in $L$ are both in $C$ and in $\pi$.
We first show, in Lemma~\ref{le:decomposition}, that vertices in $L$ appear in $C$ and $\pi$ in the same order, then, in Lemma~\ref{le:crossings}, we show that each edge in a shortest $st$-separating cycle is also contained in some single-crossing shortest $st$-separating cycle.

\begin{lemma}\label{le:decomposition}
Any multiple-crossing shortest $st$-separating cycle $C$ contains exactly one pair of vertices $a,b$ that split $C$ into two paths $C(a,b)$, $C'(a,b)$ such that:
\begin{enumerate}\itemsep0em
\item\label{nocrossing}  path $C(a,b)$ in $C$ joining $a$ and $b$ does not contain any vertex in $L$ other than $a$ and $b$;
\item\label{shortest} path $C'(a,b)$ is a shortest $ab$-path;
\item\label{order} all vertices in $L$  appear walking on $C'(a,b)$ in the same order, or the reversal, as they appear walking on $\pi$;
\item\label{not-out} all vertices in $L$  belong to $\pi(a,b)$.
\end{enumerate}
\end{lemma}
\begin{proof}
W.l.o.g., we fix a planar embedding of $G^*$ such that:
\begin{itemize}\itemsep0em
\item $\pi$ is a horizontal segment, and $f_{t}^{*}$ is its right endpoint;
\item $f_{t}^{*}$ is contained in the internal region defined by $C$, and $f_{s}^{*}$ is contained in the external region defined by $C$.
\end{itemize}
By Jordan's curve theorem, $C$ and $\pi$ have an odd number of crossings, since $C$ is a multiple-crossing $st$-separating cycle there are at least three crossings.
Since there is an odd number of crossings, walking along $C$ there must be at least two consecutive crossings $X_{1}, X_{2}$ that are both up (or both down).

We define $a$ (resp., $b$) as the internal vertex of $X_{1} \cup X_{2}$ that is closest to $f_{s}^{*}$ (resp., $f_{t}^{*}$). Obviously, both $a$ and $b$ belong to $\pi$.
Since there is at least one more crossing in addition to $X_{1}$ and $X_{2}$, we can define $C(a,b)$ as the $ab$-path in  $C$ that does not contain any crossing (thus proving point~\ref{nocrossing}), and $C'(a,b)$ as the $ab$-path in $C$ that contains all the internal vertices of  all crossings.

We observe that $C(a,b) \cup \pi(a,b)$ is an $st$-separating cycle, since it is a cycle and has only one crossing with $\pi$---it is formed by $\pi(a,b)$ plus the first (or the last) edge in $X_{1}$ and the last (or the first) edge in $X_{2}$.
 
Since $C(a,b) \cup \pi(a,b)$ is an $st$-separating cycle, and $C = C(a,b) \cup C'(a,b)$ is a shortest $st$-separating cycle, then $C'(a,b)$ cannot be longer than $\pi(a,b)$, which is a shortest $ab$-path. Hence, $C'(a,b)$ is a shortest $ab$-path, proving point~\ref{shortest}.

We prove now point~\ref{order}. It directly derives from the fact that both $\pi$ and $C'(a,b)$ are shortest paths.  
Let us consider three distinct vertices in $L$ that appear in $\pi$ in the order $x,y,z$, not necessarily consecutive, so that we have $\dist(x,y) < \dist(x,z)$ and $\dist(y,z) < \dist(x,z)$. Disregarding reversals, they may appear in $C'(a,b)$ in three possible orders: $(x,y,z)$, $(x,z,y)$, $(y,x,z)$, but order $(x,z,y)$ would imply $\dist(x,z) < \dist(x,y)$, while order $(y,x,z)$ would imply $\dist(x,z) < \dist(y,z)$, so the only possible order in $C'(a,b)$ is $(x,y,z)$ or its reversal.

Point~\ref{not-out} is a consequence of point~\ref{order}. In fact, if $C'(a,b)$ contained a vertex $w$ belonging to $V(\pi) \setminus V(\pi(a,b))$, then the order of the vertices belonging to $L$ would be different walking on $C'(a,b)$ and walking on $\pi$.

Finally, suppose that there exists a pair $a',b'$ of vertices having the same properties of the pair $a,b$ of points \ref{nocrossing} - \ref{not-out}. By point \ref{not-out}, we can say that all vertices in $L$ belong to $\pi(a,b)$ and also to $\pi(a',b')$. Since $a,b,a',b'$ belong to $L$, the pair $a',b'$ must coincide with the pair $a,b$. This proves the uniqueness of the pair $a,b$.
\end{proof}

\begin{corollary} \label{pr:uniqueness}
Given a multiple-crossing shortest $st$-separating cycle $C$, there exists exactly one pair of consecutive crossings of the same type (up crossing or down crossing) between $C$ and $\pi$.
\end{corollary}
\begin{proof}
By Lemma \ref{le:decomposition}, for each pair of consecutive crossings of the same type, there exists a unique pair of vertices $a,b$ satisfying points \ref{nocrossing} - \ref{not-out}. 
Moreover, each vertex of $L$ belongs to at most one crossing, otherwise $C$ would not be a simple cycle and then it would not be shortest. 
It follows that the vertices $a,b$ identify only one pair of consecutive crossings of the same type. Therefore, $C$ and $\pi$ have exactly one pair of consecutive crossings of the same type.
\end{proof}

The following lemma shows that in order to find all edges that belong to at least one shortest $st$-separating cycle, we can restrict our attention to single-crossing shortest $st$-separating cycles.

\begin{lemma}\label{le:crossings}
In an undirected planar graph, each edge contained in a shortest $st$-separating cycle is also contained in some single-crossing shortest $st$-separating cycle.
\end{lemma}
\begin{proof}
Let $e$ be an edge in a shortest $st$-separating cycle $C$. If $C$ is a single-crossing shortest $st$-separating cycle the thesis trivially holds. Otherwise, let us split $C$ into $C(a,b)$ and $C'(a,b)$ as defined in Lemma~\ref{le:decomposition}. 
Being $C$ a multiple-crossing $st$-separating cycle, there is an odd number of crossings, at least three, between $C$ and $\pi$.
Due to Lemma~\ref{le:decomposition}, crossings appear in the same order both in $\pi$ and in $C'(a,b)$, and only the first and the last of them are of the same type (both up crossings or both down crossings).

If $e \in C(a,b)$ or $e \in \pi(a,b)$, then $e$ is contained in the single-crossing shortest $st$-separating cycle $C(a,b) \cup \pi(a,b)$, and the thesis holds.

Otherwise, by Lemma~\ref{le:decomposition}, $e$ is contained in the portion of $C'(a,b)$ delimited by two crossings $X$ and $Y$ of different types (an up crossing and a down crossing), see Figure~\ref{fig:singlecrossing}(\subref{fig:singlecrossing_a}). Let  $D(x,y)$ be the subpath of $C'(a,b)$ such that $x,y \in \pi$ are the first internal vertices of crossings $X$ and $Y$ encountered walking from $e$ in the two directions, with no other crossings in $D(x,y)$. We observe that $D(x,y)$ consists of at least two edges, let $e_x$ and $e_y$ be the first and the last edges in $D(x,y)$, leading to $x$ and $y$.  Since $X$ is up and $Y$ is down (or vice-versa), edges $e_x$ and $e_y$ lie on the same side of $\pi$, see Figure~\ref{fig:singlecrossing}(\subref{fig:singlecrossing_b}).

\begin{figure}[h]
\captionsetup[subfigure]{justification=centering}
\centering
	\begin{subfigure}{6cm}
\begin{overpic}[width=6cm,percent]{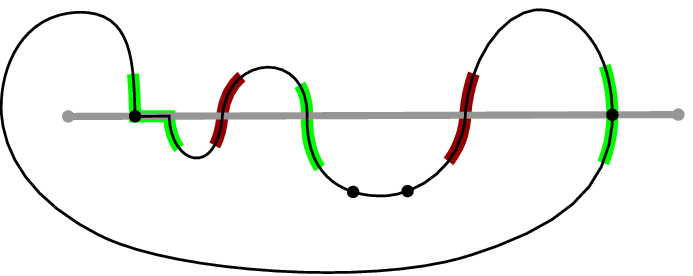}
\put(16.5,18.5){$a$}
\put(90.6,18){$b$}
\put(54,8){$e$}
\put(66.8,35){$C$}
\put(46,26){$X$}
\put(62.5,26){$Y$}
\put(79,25){$\pi$}
\end{overpic}
\caption{}\label{fig:singlecrossing_a}
\end{subfigure}	
\qquad\qquad
	\begin{subfigure}{6cm}
\begin{overpic}[width=6cm,percent]{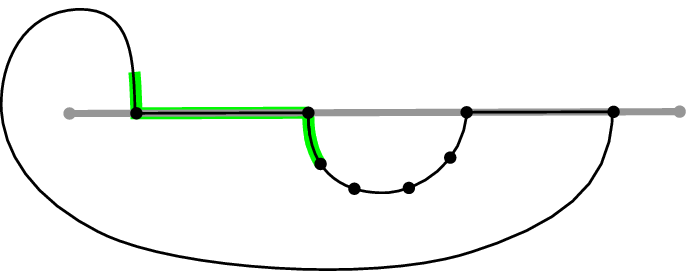}
\put(16.5,18.5){$a$}
\put(90.6,18){$b$}
\put(54,8){$e$}

\put(43.5,25){$x$}
\put(67,26){$y$}

\put(38,17){$e^x$}
\put(68,17.5){$e^y$}
\end{overpic}
\caption{}\label{fig:singlecrossing_b}
\end{subfigure}	
\caption{in (\subref{fig:singlecrossing_a}) a multiple-crossing cycle $C$, down crossings are highlighted in green and up crossings in red. Walking clockwise on $C$, $X$ is a down crossing and $Y$ is an up crossing. In (\subref{fig:singlecrossing_b}), the corresponding single-crossing $st$-separating cycle containing $e$.}
\label{fig:singlecrossing}
\end{figure}

Assume w.l.o.g.\ that $a,x,y,b$ appear in $\pi$ in this order, with possibly $a=x$ or $y=b$: path $\pi(a,x) \circ D(x,y) \circ \pi(y,b)$ is a shortest $ab$-path and does not contain any crossing with $\pi$, and cycle $C(a,b) \circ \pi(a,x) \circ D(x,y) \circ \pi(y,b)$ is a single-crossing shortest $st$-separating cycle containing edge $e$.
\end{proof}

It follows from Lemma \ref{le:crossings} that, in order to find all edges having vitality 1 we only need to refer to shortest single-crossing $st$-separating cycle in the dual graph $G^*$. Towards this goal, as in \cite{itai-shiloach,reif}, graph $G^{*}$ is ``cut'' along a shortest path $\pi$ from $f_s^{*}$ to $f_t^{*}$, obtaining the \emph{doubled dual graph} $D$, in which each vertex $f^{*}$ in $\pi$ is split into two vertices $f^{*}_{u}$ and $f^{*}_{\ell}$. Edges in $G^{*}$ incident on each $f^{*}\in V(\pi)\setminus\{f^{*}_{s},f^{*}_{t}\}$ from above $\pi$ are moved to $f^{*}_{u}$ and edges incident on $f^{*}$ from below $\pi$ are moved to $f^{*}_{\ell}$. Edges incident on $f^{*}_{s}$ or $f^{*}_{t}$ are dealt in the same way with by extending $\pi$ with two dummies edges. 
Edges on $\pi$ are copied both on $f^{*}_{u}$ and on $f^{*}_{\ell}$, so that path $\pi$ is doubled. Let $\pi_x$ be the copy of $\pi$ in $D$ formed by vertices $f^*_\ell$'s and let $\pi_y$ be the copy of $\pi$ in $D$ formed by vertices $f^*_u$'s.

It is shown in \cite{itai-shiloach,reif}, that any minimum single-crossing
$st$-separating cycle in $G^*$ corresponds in $D$ to a path from a vertex $f^*_\ell$ to a vertex $f^*_u$ having length equal to MF and viceversa. Let $k$ be the number of pairs $(f^*_\ell,f^*_u)$ such that $dist_D(f^*_\ell,f^*_u)=MF$. We denote the $i$-th pair by $(x_i,y_i)$, so that  $x_1,x_2,\ldots,x_k$ appear in this order in $\pi_x$ and $y_1,y_2,\ldots,y_k$ appear in this order in $\pi_y$.

For any $e\in E(G)$ satisfying $e^*\in\pi$ we denote by $e^D_x,e^D_y$ be the copies of $e^*$ in $D$ belonging to $\pi_x$ and $\pi_y$, respectively. Now, let $\gamma$ be an $st$-separating cycle that crosses $\pi$ exactly once. By following the same strategy of Lemma \ref{le:crossings}, if $\gamma$ uses a portion of $\pi$, then $\gamma$ corresponds in $D$ to two paths: one that uses the corresponding portion of $\pi_x$ and the other that uses the corresponding portion of $\pi_y$. Thus $e^D_x$ belongs to a shortest $x_iy_i$-path for some $i\in [k]$ if and only if $e^D_y$ belongs to a shortest $x_jy_j$-path for some $j\in [k]$.

Therefore, for our purposes we can associate to each edge of $G$ a unique edge of $D$ in the following way. If $e^*\not\in\pi$, then $e^D$ is the corresponding edge generated in $D$. Otherwise, $e^*\in\pi$ and we set $e^D=e^D_x$. The following crucial corollary is implied by the above discussion.

\begin{corollary}\label{cor:main}
Let $e$ be an edge of $G$. Then $e$ has vitality 1 if and only if $e^D$ belongs to a shortest $x_iy_i$-path in $D$, for some $i\in [k]$.
\end{corollary}

By the above corollary, an edge $e$ of $G$ has vitality 1 if and only if $e^D=ab$ is such that $dist(a,x_i)+dist(b,y_i)+1=MF$ or $dist(b,x_i)+dist(a,y_i)+1=MF$. Therefore, in order to find the edges having vitality 1 one should find in $D$ the shortest path trees rooted at each $x_i$ and each $y_i$. Since the graph is unweighted this can be done in $O(n^2)$ time by performing a BFS visit of $D$ from each $x_i$ and $y_i$. In the following we will show that such BFS's may be restricted to a subgraph of $D$ so that the overall time complexity is $O(n)$.

\subsection{Decomposing the doubled dual graph $D$}

The main result of this subsection is Theorem \ref{th:main_tilde}, that characterizes edges with vitality 1 more precisely. Let's start with some definitions.

From now on, for every $i\in [k]$ we fix an arbitrary shortest $x_iy_i$-path $p_i$, and we assume that $\{p_i\}_{i\in [k]}$ is a set of pairwise single-touch non-crossing shortest paths, where two paths $p$ and $q$ are \emph{single-touch} if $p\cap q$ is a (possibly empty) path. Moreover, we define $U=\bigcup_{i\in [k]}p_i$, see Figure \ref{fig:Left_Right_and_Omega_i}(\subref{fig:U}). Note that the single-touch property and the order in which the $x_i$'s and the $y_i$'s appear on the external face of $D$ imply that $U$ is a forest. By using the result in~\cite{balzotti-franciosa_2}, $U$ can be computed in $O(n)$ worst-case time, details are given in Subsection~\ref{sub:U}.

Each $p_i$ splits $D$ into two parts as shown in the following definition and in Figure~\ref{fig:Left_Right_and_Omega_i}(\subref{fig:Left_Right}). Let $\ell_x$ and $r_x$ be the extremal vertices of $\pi_x$ so that $\ell_x,x_1,x_k,r_x$ appear in this order in $\pi_x$ (possibly $\ell_x=x_1$ and/or $r_x=x_k$). Similarly, let $\ell_y$ and $r_y$ be the extremal vertices of $\pi_y$ so that $\ell_y,y_1,y_k,r_y$ appear in this order in $\pi_y$ (possibly $\ell_y=y_1$ and/or $r_y=y_k$).

For every $i\in [k]$, we define $\Left_i$ as the subgraph of $D$ bounded by the cycle $\pi_y\ell_y,y_i]\circ p_i\circ \pi_x[x_i,\ell_x]\circ \ell$, where $\ell$ is the leftmost $\ell_x \ell_y$-path in $D$. Similarly, we define $\Right_i$ as the subgraph of $D$ bounded by the cycle $\pi_y[y_i,r_y]\circ r \circ \pi_x[r_x,x_i]\circ p_i$, where $r$ is the rightmost $r_x r_y$-path in $D$.

\begin{definition}\label{def:Omega_i}
For every $i\in \{2,\ldots,k-1\}$, we define $\Omega_i=\Right_{i-1}\cap\Left_{i+1}$.  Moreover, we define $\Omega_1=\Left_2$ and $\Omega_k=\Right_{k-1}$ as special cases.
\end{definition}
\noindent
We note that $\Omega_i$ is the subgraph of $D$ between $p_{i-1}$ and $p_{i+1}$, see Figure \ref{fig:Left_Right_and_Omega_i}(\subref{fig:Omega_i}).

\begin{figure}[h]
\captionsetup[subfigure]{justification=centering}
\centering
\begin{subfigure}{4cm}
\begin{overpic}[width=4cm,percent]{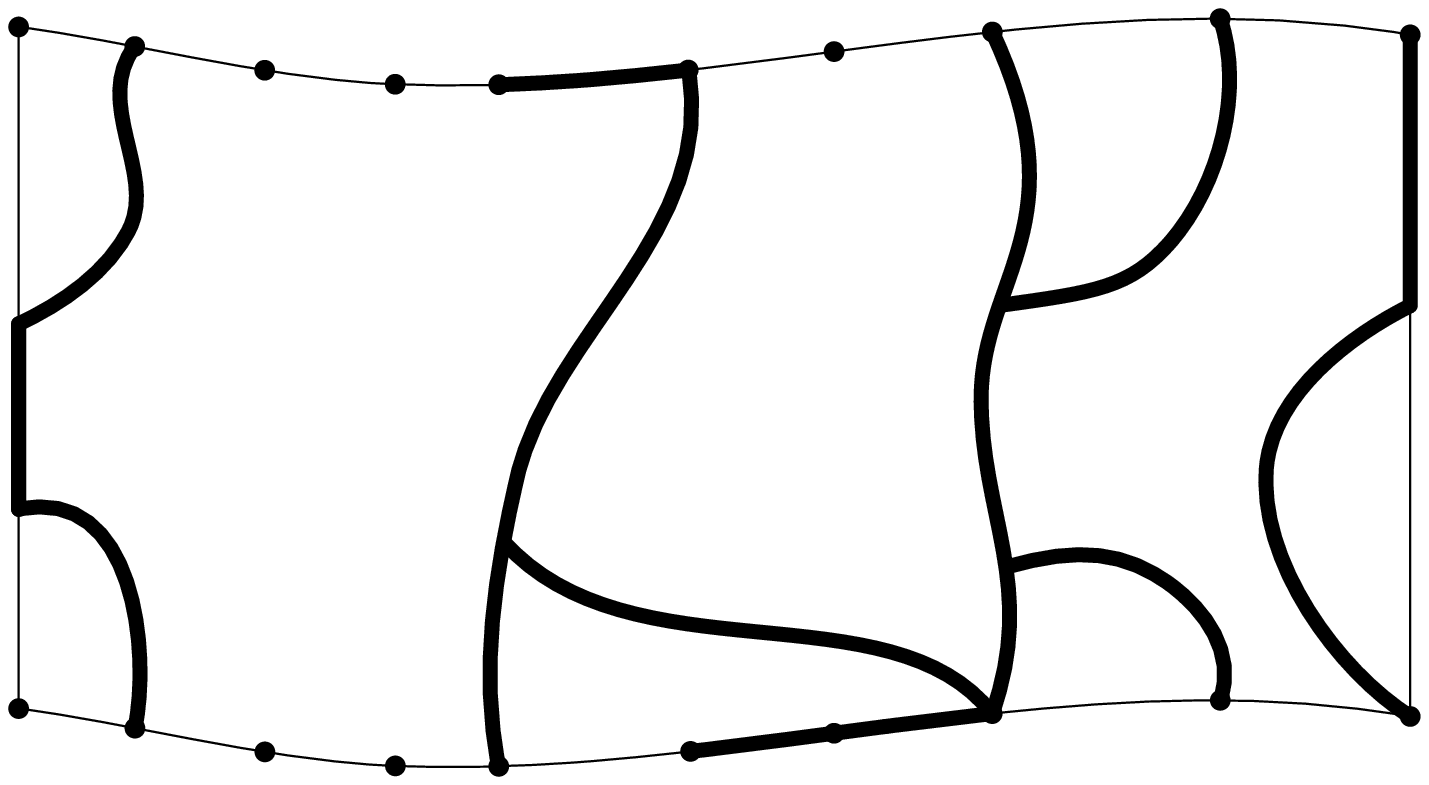}

\put(6.5,-2.9){$x_1$}
\put(8.5,55){$y_1$}
\put(33,-4.5){$x_2$}
\put(33,53){$y_2$}
\put(46.7,-3.5){$x_3$}
\put(46.7,53.6){$y_3$}
\put(67,-1.5){$x_4$}
\put(67,55.7){$y_4$}
\put(83.7,0.2){$x_5$}
\put(84,57){$y_5$}
\put(97,-1){$x_6$}
\put(97,56){$y_6$}
\end{overpic}
\vspace*{-3mm}
\caption{}\label{fig:U}
\end{subfigure}
\qquad
\begin{subfigure}{4cm}
\begin{overpic}[width=4cm,percent]{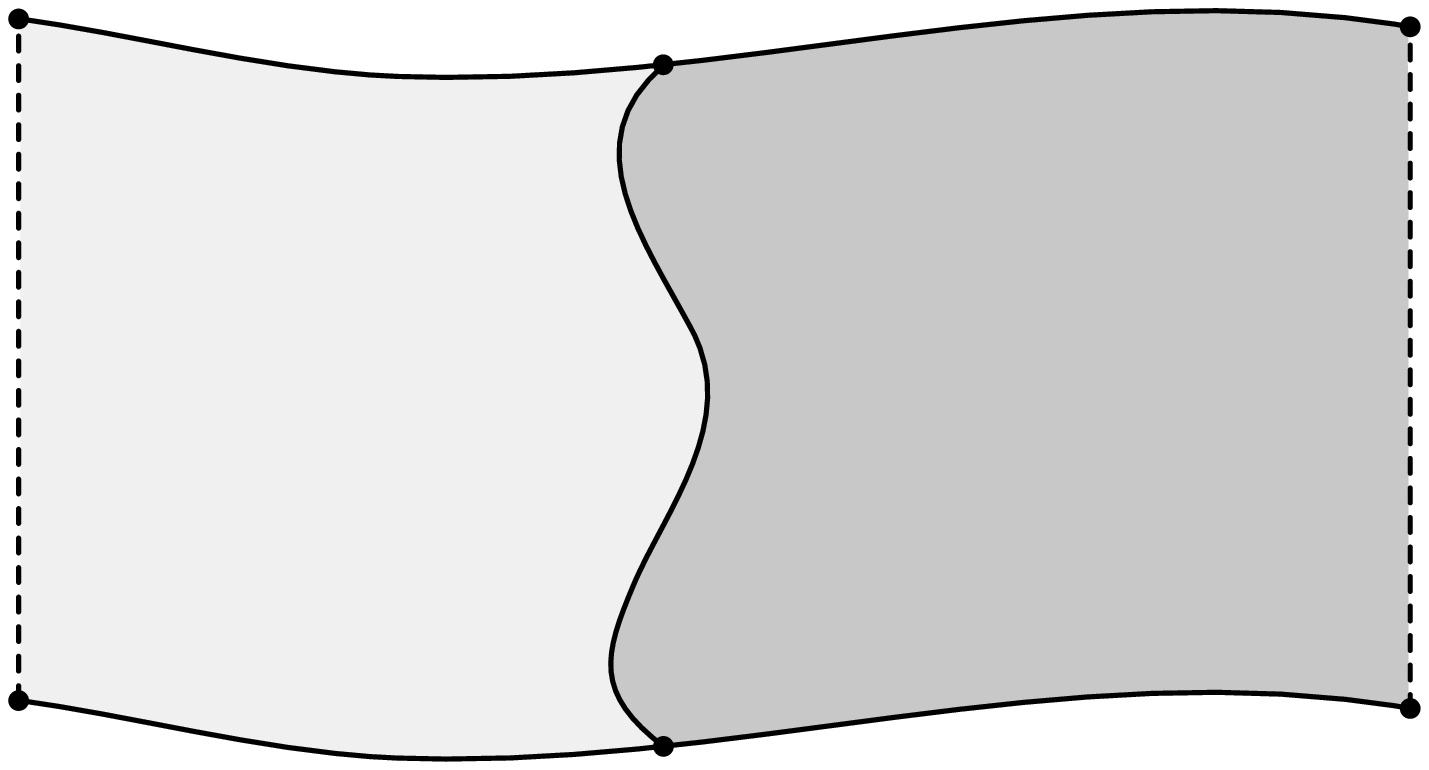}
\put(-1.5,-3.4){$\ell_x$}
\put(0,55.5){$\ell_y$}
\put(44.5,-4.3){$x_i$}
\put(45,53){$y_i$}
\put(97,-2){$r_x$}
\put(98,55){$r_y$}

\put(-5,25){$\ell$}
\put(102,25){$r$}

\put(15,25){$\Left_i$}
\put(66,25){$\Right_i$}
\end{overpic}
\vspace*{-3mm}
\caption{}\label{fig:Left_Right}
\end{subfigure}		
\qquad
\begin{subfigure}{4cm}
\begin{overpic}[width=4cm,percent]{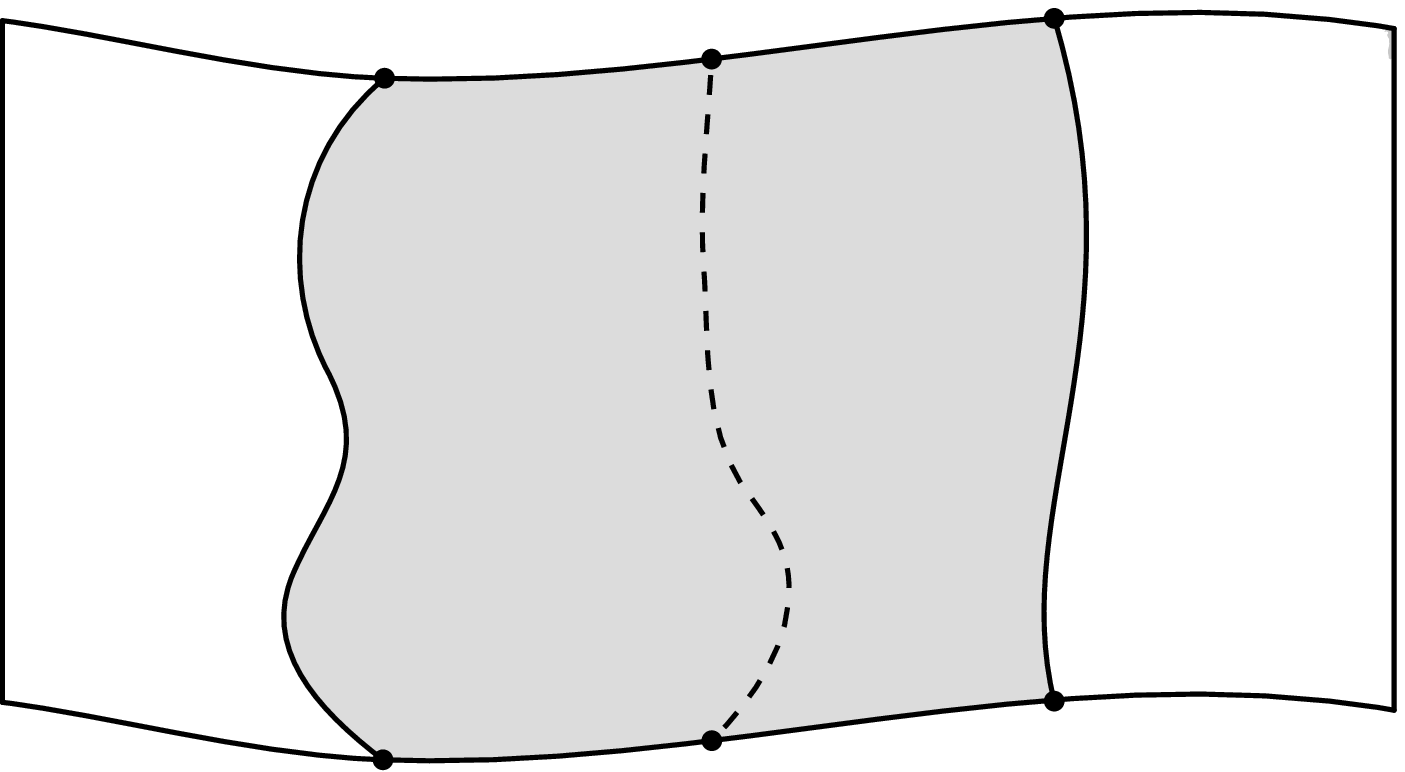}

\put(26,-5){$x_{i-1}$}
\put(26,53){$y_{i-1}$}

\put(49,-3.5){$x_i$}
\put(49.5,54){$y_i$}

\put(73,-1){$x_{i+1}$}
\put(73,57.3){$y_{i+1}$}

\put(6,25){$p_{i-1}$}
\put(42,25){$p_i$}
\put(78,25){$p_{i+1}$}

\put(59,38){$\Omega_i$}

\end{overpic}
\vspace*{-3mm}
\caption{}\label{fig:Omega_i}
\end{subfigure}		
\caption{(\subref{fig:U}) graph $U$ in bold, (\subref{fig:Left_Right}) subgraphs $\Left_i$ and $\Right_i$, (\subref{fig:Omega_i}) subgraph $\Omega_i$, for some $i\in [k]$.}
\label{fig:Left_Right_and_Omega_i}
\end{figure}

\noindent
The following lemma 
is an intermediate step to reach the main result stated in Theorem \ref{th:main_tilde}, showing that we can restrict our attention to shortest paths contained in $\Omega_i$, for $i\in[k]$.

\begin{lemma}\label{lemma:e_in_Omega_j}
An edge $e\in E(G)$ has vitality 1 if and only if $e^D$ belongs to a shortest $x_iy_i$-path in $\Omega_i$, for some $i\in [k]$.
\end{lemma}
\begin{proof}
By Corollary \ref{cor:main}, it suffices to prove that an edge $e^D\in E(D)$ belongs to a shortest $x_iy_i$-path in $D$, for some $i\in [k]$, if and only if $e^D$ belongs to a shortest $x_jy_j$-path in $\Omega_j$, for some $j\in [k]$.

The ``if part" is trivial, thus it suffices to prove the ``only if" part. Let $\rho$ be a shortest path in $D$ from $x_i$ to $y_i$, for some $i\in [k]$, such that $e^D\in\rho$. We have to prove that there exists a shortest path $\sigma$ from $x_j$ to $y_j$, for some $j\in [k]$, satisfying $e^D\in E(\sigma)$ and $\sigma\subseteq\Omega_j$.

Let us assume that $\rho\not\subseteq\Omega_i$, otherwise the thesis holds by taking $j=i$. There are three cases: either $e^D\in\Omega_1$, or $e^D\in\Omega_k$, or $e^D\in\Omega_j\cap\Omega_{j+1}$ for some $j\in \{1,\ldots,k-1\}$. We study only the last case, the others being similar.

W.l.o.g., we assume that $i\leq j$. We want to find a shortest path $\sigma$ from $x_j$ to $y_j$ containing all edges belonging to $\rho\cap(\Omega_j\cap\Omega_{j+1})$. Being $i\leq j$ and $e^D\in\Right_j$, then $\rho$ crosses $p_j$ a positive and even number of times. Let $a,b\in V(p_j)\cap V(\rho)$ be the first and last vertex, respectively, on $p_j$ going from $x_j$ to $y_j$, possibly $a=x_j$ and/or $b=y_j$ (see Figure~\ref{fig:serpentina}). Thus the path $\sigma=p_j[x_j,a]\circ\rho[a,b]\circ p_j[b,y_j]$ is a shortest path because $\rho$ and $p_j$ are. If $\sigma\subseteq\Omega_j\cap\Omega_{j+1}=\Right_j\cap\Left_{j+1}$, then the thesis holds. Otherwise $\sigma$ crosses both $p_j$ and/or $p_{j+1}$, and to complete the proof it suffices to replace minimal subpaths in $\Left_j\setminus\Right_j$ by subpaths of $p_j$ with the same extremal vertices and to replace minimal subpaths in $\Right_{j+1}\setminus\Left_{j+1}$ by subpaths of $p_{j+1}$ with the same extremal vertices as shown in Figure \ref{fig:serpentina}; these replacements are univocally determined by the same argument as in Lemma \ref{le:crossings}. After this, $\sigma$ is a shortest $x_jy_j$-path that intersects both $p_j$ and $p_{j+1}$ without crossings, thus $\sigma\subseteq\Omega_j\cap\Omega_{j+1}$.
\end{proof}

\begin{figure}[h]
\captionsetup[subfigure]{justification=centering}
\centering
	\begin{subfigure}{10.5cm}
\begin{overpic}[width=10.5cm,percent]{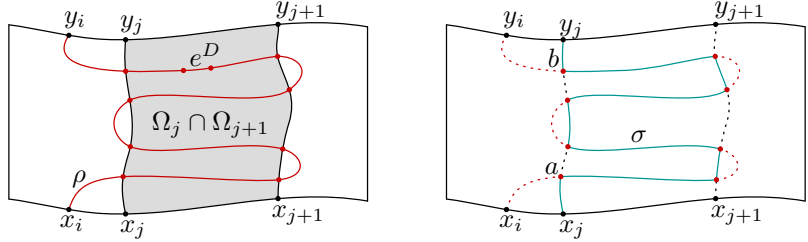}
\put(6.5,-1.2){$x_i$}
\put(7,24.5){$y_i$}

\put(14,-1.8){$x_j$}
\put(14,24){$y_j$}

\put(33,0){$x_{j+1}$}
\put(33.5,25.8){$y_{j+1}$}

\put(8,4.5){$\rho$}
\put(18,11){$\Omega_j\cap\Omega_{j+1}$}
\put(22.7,19.1){$e^D$}

\put(61.5,-1.2){$x_i$}
\put(62,24.5){$y_i$}

\put(69,-1.8){$x_j$}
\put(69,23.5){$y_j$}

\put(88,0){$x_{j+1}$}
\put(88.5,25.8){$y_{j+1}$}

\put(78,9.5){$\sigma$}

\put(67.2,5.5){$a$}
\put(67.6,19){$b$}

\end{overpic}
\end{subfigure}		
\caption{on the left a shortest $x_iy_i$-path $\rho$ in $D$. On the right a shortest $x_jy_j$-path $\sigma$ in $\Omega_j$.}
\label{fig:serpentina}
\end{figure}

\noindent
By Lemma \ref{lemma:e_in_Omega_j}, the search of edges in shortest $x_iy_i$-paths in $D$ required in Corollary \ref{cor:main} can be restricted to $\Omega_i$, for all $i\in[k]$. Despite this is a simplification, an edge $e^D\in E(D)$ can still belong to many $\Omega_i$'s as shown in Figure~\ref{fig:clessidra}. Thus the total size of all $\Omega_i$'s might be $\Theta(n^2)$. To overcome this problem we will introduce a compact representation of the $\Omega_i$'s---see Definition~\ref{def:omega_tilde}---whose dimension is linear.

As a preliminary step, in the following lemma we show that every edge belonging to more than two $\Omega_i$'s belongs to $U$.

\begin{figure}[h]
\captionsetup[subfigure]{justification=centering}
\centering
	\begin{subfigure}{5cm}
\begin{overpic}[width=5cm,percent]{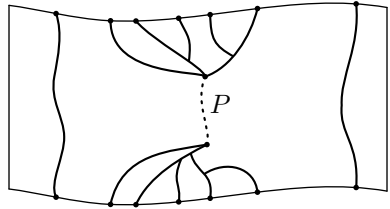}

\put(53,25){$P$}

\end{overpic}
\end{subfigure}		
\caption{in bold edges of $U$, in dotted line a set of edges $P$ belonging to many $\Omega_i$'s.}
\label{fig:clessidra}
\end{figure}

\begin{lemma}\label{lemma:i,j,l}
Let $i,j\in [k]$ satisfy $i+2\leq j$ and let $e^D$ be an edge of $D$, then
\begin{equation*}
e^D\in\Omega_i\cap\Omega_j\Longleftrightarrow e^D\in p_z \text{ for all }z\in[i+1, j-1].
\end{equation*}
\end{lemma}
\begin{proof}
The single-touch property implies that $e^D\in p_z \text{ for all }z\in[i+1, j-1]$ if and only if $e^D\in p_{i+1}\cap p_{ j-1}$. Moreover, if $e^D\in p_z$ for all $z\in[i+1, j-1]$, then $e^D\in\Omega_i\cap\Omega_j$ by definition of the $\Omega_\ell$'s. To complete the proof, it suffices to prove that $e^D\in\Omega_i\cap\Omega_j\Longrightarrow e\in p_{i+1}\cap p_{ j-1}$.

Let us assume by contradiction that $e^D\not\in p_{i+1}\cap p_{ j-1}$ and $e\in\Omega_i\cap\Omega_j$. W.l.o.g., we assume that $e^D\not\in p_{i+1}$. There are two cases: either $e^D\in\Right_{i+1}\setminus p_{i+1}$ or $e^D\in\Left_{i+1}\setminus p_{i+1}$. In the first case $e^D\not\in\Left_{i+1}$ and thus $e^D\not\in\Omega_i$, absurdum. In the second case, $e^D\in\Left_{ j-1}\setminus p_{ j-1}$ and thus $e^D\not\in\Omega_{ j}$, absurdum.
\end{proof}

\noindent
After Lemma \ref{lemma:i,j,l}, we can introduce a compact representation of the $\Omega_i$'s.

\begin{definition}\label{def:omega_tilde}
For all $i\in \{2,\ldots,k-1\}$, if $p_{i-1}\cap p_{i+1}$ is a non-empty path $q_i$, then we define $\widetilde{\Omega}_i$ as $\Omega_i$ in which we replace all edges of $q_i$ by a single edge of length $|q_i|$. Otherwise, we define $\widetilde{\Omega}_i=\Omega_i$. Moreover, we define $\widetilde{\Omega}_1=\Omega_1$ and $\widetilde{\Omega}_k=\Omega_k$.
\end{definition}

\noindent
Now we can state the main result, which characterizes edges with vitality 1 more precisely. Indeed, while an edge of $D$ might belong to many $\Omega_i$'s, it belongs either to $U$ or to at most two $\widetilde{\Omega}_i$'s.

\begin{theorem}\label{th:main_tilde}
Let $e\in E(G)$,
\begin{itemize}\itemsep0em
\item if $e^D\in U$, then $e$ has vitality 1,
\item otherwise, there exists a unique $i\in[k]$ satisfying $e^D\in\widetilde{\Omega}_i\cup\widetilde{\Omega}_{i+1}$, moreover, $e$ has vitality 1 if and only if $e^D$ belongs to a shortest $x_iy_i$-path in $\widetilde{\Omega}_i$ and/or $e^D$ belongs to a shortest $x_{i+1} y_{i+1}$-path in $\widetilde{\Omega}_{i+1}$. 
\end{itemize}
\end{theorem}
\begin{proof}
The first statement is implied by Corollary \ref{cor:main} and the definition of $U$. To prove the second statement let us assume $e^D\not\in U$. First we prove that there exists a unique $i\in[k]$ satisfying $e^D\in\Omega_i\cup\Omega_{i+1}$. By Definition \ref{def:Omega_i} there are three cases: $e^D\in\Omega_1$, $e^D\in\Omega_k$ and $e^D\in\Omega_j\cap\Omega_{j+1}$ for some $j\in\{1,\ldots,k-1\}$. Lemma \ref{lemma:i,j,l} and $e^D\not\in U$ imply that there are no other cases. If $e^D\in\Omega_1$, then $i=1$; similarly, if $e^D\in\Omega_k$, then $i=k-1$. Finally, if $e^D\in\Omega_j\cap\Omega_{j+1}$, then $i=j$; indeed $e^D\not\in\Omega_j$ for all $j\not\in\{i,i+1\}$ because of Lemma \ref{lemma:i,j,l} and $e^D\not\in U$. Therefore, in all cases, $e^D\in\widetilde{\Omega}_i\cup\widetilde{\Omega}_{i+1}$ and $e^D\not\in\widetilde{\Omega}_j$ for all $j\not\in\{i,i+1\}$ by Definition \ref{def:omega_tilde}.

We note that $e^D$ belongs to a shortest $x_i y_i$-path in $\Omega_i$ if and only if $e^D$ belongs to a shortest $x_i y_i$-path in $\widetilde{\Omega}_i$ because when switching from $\Omega_i$ to $\widetilde{\Omega}_i$ we have replaced all edges in $q_i$ (see Definition \ref{def:omega_tilde}) with a path of the same length. We note that all the internal vertices of $q_i$ have degree two in $\Omega_i$ and thus any path whose extremal vertices are not in $q_i$ contains either all edges or no edges of $q_i$. Thus the thesis follows from Lemma  \ref{lemma:e_in_Omega_j}. 
\end{proof}

\section{A linear time algorithm for computing the vitality of all edges}\label{se:algo}

In this section we show that the max-flow vitality of all edges of an unweighted undirected planar graph can be found in linear time, see Theorem \ref{th:main_vitality_edge_is_linear}. By following the result in Theorem \ref{th:main_tilde}, it suffices to show that $U$ and all the $\widetilde{\Omega}_i$'s can be computed in linear time.

\subsection{Computing $U$}\label{sub:U}

Balzotti and Franciosa \cite{balzotti-franciosa_2} show how to compute the union of a set of non-crossing single-touch shortest paths with extremal vertices $x_i,y_i$ on the external face in an undirected unweighted planar graph in linear time. They solve the general case in which pairs of extremal vertices can be nested arbitrarily, provided that for  any two pairs $(x_i,y_i)$ and $(x_j,y_j)$ they do not appear on the infinite face in the order  $x_i,x_j,y_i,y_j$. This property ensures that non-crossing paths connecting each pair can be found.

In our case the vertices of the external face appear in the order $x_1,x_2,\ldots,x_r,y_r,y_{r-1},\ldots,$ $y_2,y_1$. This fact, together with the single-touch property, implies that the union of the shortest paths is a forest. 

\begin{theorem}[\cite{balzotti-franciosa_2}]\label{th:U_complexity}
Given $D$, we can compute $U=\bigcup_{i\in [k]}p_i$ in $O(n)$ worst-case time, where $p_i$ is a shortest $x_iy_i$-path and $U$ is a forest.
\end{theorem}

\subsection{Computing the $\widetilde{\Omega}_i$'s}

We show how to compute $\widetilde{\Omega}_i$, for all $i\in[k]$. We use the results of Gabow and Tarjan \cite{gabow-tarjan} that allows us to compute the nearest common ancestor (\texttt{nca}($\cdot,\cdot$)) of $m$ pairs of vertices in a tree with $n$ vertices in $O(n+m)$ worst-case time.

\begin{lemma}\label{lemma:Henzinger_in_Omega_i}
It holds that $\sum_{i\in[k]}|E(\widetilde{\Omega}_i)|=O(n)$. Moreover, given $U$, we can compute $\widetilde{\Omega}_i$ for all $i\in [k]$ in $O(n)$ total worst-case time.
\end{lemma}
\begin{proof}
Let's start proving the first statement. Let $e^D\in E(D)$ belong to $\Omega_i$ for all $i\in I$, for some subset $I\subseteq [k]$. We prove that $e^D$ belongs to at most two $\widetilde{\Omega}_i$'s. If $|I|\leq2$, then the thesis holds. Otherwise,  Lemma~\ref{lemma:i,j,l} implies that $I$ is in the form $I=\{\ell,\ldots,\ell+z\}$ and $e^D\in U$. Hence, by Definition~\ref{def:omega_tilde}, $e^D$ is contracted in all $\widetilde{\Omega}_i$ for $i\in\{\ell+1,\ldots,\ell+z-1\}$ and $e^D$ appears only in $\widetilde{\Omega}_\ell$ and $\widetilde{\Omega}_{\ell+z}$. Thus $\sum_{i\in[k]}|E(\widetilde{\Omega}_i)|=O(n)+O(k)$ because in any $\widetilde{\Omega}_i$'s all contracted edges are replaced by only one edge. Being $k<n$, then $\sum_{i\in[k]}|E(\widetilde{\Omega}_i)|=O(n)$ as we claimed.

Now we prove that we can compute $\widetilde{\Omega}_i$ for all $i\in [k]$ in $O(n)$ total worst-case time. For all $i\in\{2,\ldots,k-1\}$ if $p_{i-1}$ and $p_{i+1}$ intersect, then their intersection is a path because of the single-touch property, and it is contained in $p_i$ by Lemma \ref{lemma:i,j,l}. Thus if $p_{i-1}$ and $p_{i+1}$ intersect, then we need their intersection and its length to compute $\widetilde{\Omega}_i$ in $O(|E(\widetilde{\Omega}_i)|)$ by a BFS visit.

We recall that $U$ is a forest. We compute $O(k)$ nearest common ancestors and we compare them to compute the required intersection paths. Indeed, the extremal vertices of the wanted intersection paths are always among the computed nearest common ancestors as shown in the sequel.

Each tree $T$ in $U$ is the union of a set of paths $p_i$ from $x_i$ to $y_i$: we root each $T$ in vertex $x_{\ell(T)}$, where $\ell(T)$ is the smallest index such that $x_{\ell(T)} \in T$.

For each $i\in\{1,\ldots,k-1\}$, we compute the intersection path between $p_i$ and $p_{i+1}$. We know that if $p_{i}$ and $p_{i+1}$ are in the same tree $T$, then their vertex intersection is not empty. We use all ${4 \choose 2}$ nearest common ancestors between all pairs generated by the extremal vertices of $p_i$ and $p_{i+1}$.  After rooting $T$ in $\ell(T)$, a few configurations are possible that can be distinguished by comparing all ${6 \choose 2}$ possible pairs obtained from the  previous ${4 \choose 2}$ nearest common ancestors. One configuration is shown in Figure~\ref{fig:Gabow_Tarjan}, the other two occur, informally speaking, when the first vertex in both $p_i$ and $p_{i+1}$ lies above $b$ or below $a$, referring to Figure~\ref{fig:Gabow_Tarjan}. 

For each $i\in\{2,\ldots,k-1\}$, we compute the intersection path between $p_{i-1}$ and $p_{i+1}$ (if it exists) by using the ${4 \choose 2}$  nearest common ancestors between all pairs generated by the extremal vertices of $p_{i-1}\cap p_i$ and $p_i\cap p_{i+1}$ (if they exist) and by comparing them as above. 

By Gabow and Tarjan's result \cite{gabow-tarjan}, we spend overall $O(|V(U)|+k)=O(n)$ worst-case time, because we require ${4 \choose 2}$  nearest common ancestor queries for each pair $i,i+1$ and ${4 \choose 2}$  nearest common ancestor queries for each pair $i-1,i+1$. After a proper labeling of vertices in $T$ obtained by a BFS visit starting from $x_{\ell(T)}$, by knowing the extremal vertices of $p_{i-1}\cap p_{i+1}$ we obtain its length in $O(1)$ time. Thus we compute $\widetilde{\Omega}_i$ in $O(|E(\widetilde{\Omega}_i)|)$.
\end{proof}

\begin{figure}[h]
\captionsetup[subfigure]{justification=centering}
\centering
	\begin{subfigure}{5cm}
\begin{overpic}[width=5cm,percent]{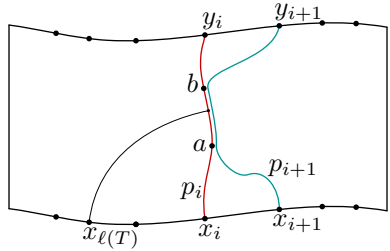}
\put(19.5,-3){$x_{\ell(T)}$}
\put(50,-2){$x_{i}$}
\put(69.5,0.4){$x_{i+1}$}

\put(51,54){$y_{i}$}
\put(70,56.5){$y_{i+1}$}

\put(45.5,9){$p_{i}$}
\put(68.5,15){$p_{i+1}$}

\put(48.5,20.5){$a$}
\put(47,36){$b$}


\end{overpic}
\end{subfigure}		
\caption{in this configuration $p_i\cap p_{i+1}=p_i[a,b]$. It holds that $a=\nca(x_i,x_{i+1})$ and $b=\nca(y_i,y_{i+1})$. This is the only  configuration in which $\nca(x_i,y_i)=\nca(x_{i+1},y_{i+1})$.}
\label{fig:Gabow_Tarjan}
\end{figure}

\subsection{Computing the vitality of all edges in linear time}

We have to determine, for each edge in $\widetilde{\Omega}_i$, whether it belongs to a shortest $x_iy_i$-path. To accomplish this, we perform BFS visits from $x_i$ and $y_i$ to vertices in $\widetilde{\Omega}_i$. Through BFS visits, we can answer the above question in constant time as showed in the following remark.

\begin{remark}\label{remark:distances_in_Omega_j}
Let $e^D$ be an edge in $\widetilde{\Omega}_i$, for some $i\in [k]$. Then $e^D$ belongs to a shortest $x_iy_i$-path in $\widetilde{\Omega}_i$ if and only if $d_{\widetilde{\Omega}_i}(x_i,a)+d_{\widetilde{\Omega}_i}(y_i,b)+1=MF$ or $d_{\widetilde{\Omega}_i}(x_i,b)+d_{\widetilde{\Omega}_i}(y_i,a)+1=MF$, where $a,b\in V(G)$ are the extremal vertices of $e^D$.
\end{remark}

\noindent
The process described in this paper is summarized in algorithm \texttt{EdgeVitality} which computes the $st$ max-flow vitality for all edges of the input graph. Its correctness follows from Theorem~\ref{th:main_tilde} and we prove its complexity in Theorem~\ref{th:main_vitality_edge_is_linear}.

\begin{figure}[h]
\begin{algorithm}[H] 
\SetAlgorithmName{Algorithm \texttt{EdgeVitality}}{}{}
\renewcommand{\thealgocf}{}
 \caption{}
 \KwIn{a undirected unweighted planar graph $G$, a source $s$ and a sink $t$}
 \KwOut{the $st$ max-flow vitality for each edge $e\in E(G)$}
{compute $G^*$, $\pi$ and $D$\label{line:1}\;
compute $U$ by using the algorithm in \cite{balzotti-franciosa_2}\label{line:2}\;
\For{$i \in [k]$}{compute $\widetilde{\Omega}_i$ by using Gabow and Tarjan's algorithm as explained in Lemma \ref{lemma:Henzinger_in_Omega_i}\;
         visit  $\widetilde{\Omega}_i$ by two BFSs from $x_i$ and $y_i$\;
            	\For{$v \in V(\widetilde{\Omega}_i)$}{
      	    compute $d_{\widetilde{\Omega}_i}(x_i,v)$ and $d_{\widetilde{\Omega}_i}(y_i,v)$\label{line:7}\;
      	   }
      	   }
\For{\label{line:8}$e\in E(G)$}{
	use Theorem \ref{th:main_tilde} and Remark \ref{remark:distances_in_Omega_j} to compute $vit(e)$\;
	}
}
\end{algorithm}
\end{figure}

\begin{theorem}\label{th:main_vitality_edge_is_linear}
Let $G$ be an unweighted undirected planar graph, and let $s,t\in V(G)$. Algorithm \texttt{EdgeVitality} computes the $st$ max-flow vitality of all edges in $O(n)$ worst-case time.
\end{theorem}
\begin{proof}
Line~\ref{line:1} can be executed in $O(n)$ worst-case time. Lines~\ref{line:2} through~\ref{line:7} can be executed in $O(n)$ worst-case time by Theorem \ref{th:U_complexity}, Lemma \ref{lemma:Henzinger_in_Omega_i} and the linear time algorithm for BFS visits in~\cite{henzinger}. Finally, for each edge $e\in E(G)$ the cycle in Line~\ref{line:8} requires $O(1)$ time.
\end{proof}

\section{Conclusions and further work}
\label{se:conclusions}

In this paper we have shown how to efficiently compute the vitality of all edges with respect to $st$ max-flow, in the case of undirected unweighted planar graphs. Note that our result also holds for uniform edge weights. In a previous work Ausiello \emph{et al.}~\cite{ausiello-franciosa_1} presented algorithms for solving the vitality problem for all edges, in the case of general undirected graphs, and for all arcs and nodes, plus contiguous arc sets, for $st$-planar graphs, both directed and undirected.

Several points remain open, most notably the case of planar weighted undirected graphs, for which it is not sufficient to look for minimum $st$-separating cycles in the dual graph, hence it is not possible to rely on the divide and conquer technique by Reif. The technique described in this paper cannot be applied to planar directed graphs either, even in the unweighted case, due to the fact that it is not always possible to avoid crossings between shortest $st$-separating cycles.


\bibliographystyle{siam}
\bibliography{biblio.bib}

\end{document}